\documentclass[letterpaper,journal]{IEEEtran}

\usepackage[T1]{fontenc}
\usepackage[utf8]{inputenc}
\usepackage{amsmath,amssymb,amsfonts}
\usepackage{amsthm}
\usepackage{booktabs}
\usepackage{tabularx}
\usepackage{array}
\usepackage{url}
\usepackage{cite}
\usepackage{hyperref}
\usepackage{tikz}
\usepackage{pgfplots}
\pgfplotsset{compat=1.18}
\usetikzlibrary{arrows.meta,positioning,fit,backgrounds}

\newcommand\submittedtext{%
  \footnotesize This work has been submitted to the IEEE for possible publication. Copyright may be transferred without notice, after which this version may no longer be accessible.}

\newcommand\submittednotice{%
\begin{tikzpicture}[remember picture,overlay]
\node[anchor=south,yshift=10pt] at (current page.south) {\fbox{\parbox{\dimexpr0.65\textwidth-\fboxsep-\fboxrule\relax}{\submittedtext}}};
\end{tikzpicture}%
}

\newcommand{\Psat}{P_{\mathrm{sat}}} 
\newcommand{\Gmin}{G_{\psi}^{\min}}

\newcommand{\GrefLoo}{G_{\psi}^{\mathrm{ref,loo}}}
\newcommand{\Vpsi}{V_{\psi}}

\newtheorem{theorem}{Theorem}
\newtheorem{lemma}{Lemma}

\newtheorem{definition}{Definition}

\title{Satisfaction Is Not Explanation: Auditing Vacuity and Training Influence in Temporal-Logic-Guided Reinforcement Learning}
\author{Lorenzo~Bacchiani}

\begin{document}

\maketitle
\submittednotice

\begin{abstract}
A reinforcement learning policy that satisfies its temporal-logic
specification has passed a test, not an assurance argument. The clause that
matters to a reviewer may never have mattered to training: it may have been
avoided entirely, forced by the environment regardless of what the policy
learned, or redundant next to the ordinary task reward. Satisfaction
probability and task return cannot tell any of this apart.

This paper introduces an audit layer that can. It measures whether a
specification clause was actually exercised, whether that role was forced or
chosen, and whether the obvious way to test causation, weakening the clause
and retraining, is even valid. Often it is not: we prove that comparable
weaker/stronger training objectives can share perfect optima under standard
acceptance-derived rewards, show related ablation hazards across a large
corpus of published specifications, and then show that a properly designed
intervention detects the effect it should. Across standard reinforcement learning benchmarks and
published external artifacts, the audit layer separates six regimes that a
single satisfaction number collapses into one. A policy that satisfies its
specification has answered whether. This paper asks why.

\end{abstract}


\section{Introduction}

Suppose a reinforcement-learning policy is deployed with a formal guarantee:
it satisfies its temporal-logic specification with probability one. A
practitioner is tempted to read this as an assurance argument. That reading is
not always available. The policy may never encounter the situation the
specification addresses. The environment may make the specification's
outcome unavoidable no matter how the policy behaves. The reward the policy
was actually trained on may already produce the same behavior on its own, so
the specification never shaped training at all. In every one of these cases
the policy genuinely satisfies the formula, and satisfaction probability
alone cannot tell a reviewer which case occurred. 
This gap matters because temporal logic already plays several different
roles inside a reinforcement-learning pipeline. A formula can define a task
evaluated only at the end of training, source a sparse terminal reward,
shape a continuous reward signal, or sit beside the learner as a monitor
that never touches the reward at all. The standard empirical summary,
satisfaction probability alongside ordinary task return, treats these roles
as interchangeable. They are not: a policy trained under one integration
mechanism can satisfy the same formula as a policy trained under another
while relying on the specification in a completely different way, or not
relying on it at all.

Our work develops an audit layer that asks, for a learned policy and a
specification clause, what happened to that clause in the behavior the
policy learned. The audit is
organized around three questions:
\begin{itemize}
\item \emph{AQ1: did the learned policy satisfy
the formula?} This is the ordinary success question, and if the answer is no
then the right explanation is failure or learnability, not vacuity.
\item \emph{AQ2: did the audited clause matter on the successful traces?} This
asks whether the clause was actually exercised once the policy succeeded,
or whether satisfaction would have held without it.
\item \emph{AQ3: what explains that clause status?} A clause may be
irrelevant because the environment forced it, because the learner selected a
vacuous policy among matched alternatives, because satisfying it trades off
against native task reward, because the chosen mechanism failed to learn a
satisfying policy, or because the proposed training-intervention test is
invalid. These questions separate outcomes that satisfaction and return
alone conflate.
\end{itemize}

On the last question, the obvious experiment is to weaken the clause,
retrain, and compare policies. That test is often invalid. We prove that
weakened and unweakened training objectives can have overlapping perfect
optima, so any observed difference may be nothing but tie-breaking. A
corrected experiment, comparing two genuinely independent objectives
instead, confirms that training influence is measurable once the test itself
is valid. We evaluate the audit layer on standard reinforcement-learning benchmarks
drawn from the literature~\cite{dietterich2000maxq,icarte2018reward} and on
published external artifacts this paper did not design: a benchmark for
shielded safe reinforcement learning~\cite{alshiekh2018shielding} and a suite
of gridworlds built to expose specification failure
modes~\cite{leike2017aisafetygridworlds}. We additionally audit a corpus of
173 published temporal-logic reinforcement-learning
specifications~\cite{jackermeier2025deepltl,icarte2018reward} to check
whether the hazard behind our validity screen is common in practice rather
than an artifact of our own formulas.

The contributions of this work can be summarized as follows: 
\begin{itemize}
  \item An audit protocol that classifies, for a specification clause, why
  it holds or fails through six regimes that satisfaction and return alone
  cannot distinguish.
  \item Exact, checkable diagnostics for conditional vacuity, environmental
  floors, and matched-reference excess, so a regime label is a measurement
  rather than a narrative choice made after the fact.
  \item A theorem proving that a common training-influence test, weakening a
  clause and retraining, can fail to detect an effect even when the training
  mechanism genuinely matters, together with a corrected intervention that we
  confirm does detect it.
  \item An evaluation across standard benchmarks, published external
  artifacts, and a large corpus of published specifications, showing that
  satisfaction-only reporting hides distinctions that matter before a
  specification-guided policy is deployed.
\end{itemize}

Code and data for the audit layer and all reported experiments are publicly
available at~\cite{bacchiani2026artifact}.

The rest of the paper is organized as follows. \S\ref{sec:background}
fixes the decision-making model, the logic used to specify behavior over it,
and the construction that connects the two. \S\ref{sec:running-example}
works a concrete case by hand, so that the rest of the paper builds on
an example already understood.
\S\ref{sec:related-work} situates the audit layer against prior work on
temporal-logic-guided reinforcement learning, formal vacuity, and mutation
testing.
\S\ref{sec:methodology} maps the audit questions to reported regimes,
defines the exact diagnostics the audit layer computes, and proves the
training-influence validity screen.
\S\ref{sec:benchmarks} describes the benchmark and artifact corpus.
\S\ref{sec:results} reports the regimes observed across the corpus,
including the corrected training-influence intervention.
\S\ref{sec:practitioner-guidance} explains how practitioners should respond
to each regime. \S\ref{sec:discussion-limitations} examines threats to validity, and
\S\ref{sec:conclusion} concludes.
 
\section{Background}
\label{sec:background}

This section fixes vocabulary and notation used throughout the paper: the
decision-making model a policy is trained in, the logic used to write a
specification over it, and the construction that lets the two be evaluated
together.

\subsection{Markov Decision Processes}
\label{subsec:mdp-background}

A finite-horizon Markov decision process (MDP)~\cite{puterman1994mdp} is a
tuple $M=(S,A,P,r,H,d_0,L)$. $S$ is a finite set of states and $A$ is a
finite set of actions. $P(s'\mid s,a)$ is a transition probability
distribution over successor states given a state-action pair. $r$ is a
reward function, $H$ is a finite horizon, and $d_0$ is a distribution over
initial states. $L$ is a labeling function that maps each state to a set of
atomic propositions, $L(s)\subseteq AP$, drawn from a fixed vocabulary $AP$;
labels are the only channel through which anything outside the MDP, in
particular the temporal-logic specifications introduced in
\S\ref{subsec:ltlf-preliminaries}, can refer to what happens inside it. An
episode is a sequence $s_0,a_0,s_1,a_1,\ldots,s_H$ generated by sampling
$s_0\sim d_0$ and, at every step $t<H$, an action $a_t$ from a policy and a
successor $s_{t+1}\sim P(\cdot\mid s_t,a_t)$. A policy $\pi$ maps a state,
together with a time index whenever the horizon is finite, to a distribution
over actions. The standard optimization criterion is expected return,
\[
J(\pi)=\mathbb{E}_{\tau\sim\pi}\Big[\textstyle\sum_{t=0}^{H-1}
r(s_t,a_t,s_{t+1})\Big],
\]
the expectation being over episodes $\tau$ that $\pi$ generates in $M$. Every
MDP in this paper is undiscounted: because $H$ is already finite, the return
stays bounded without a discount factor, so this is written as $\gamma=1$
wherever a discount factor is referenced elsewhere in the paper. Every MDP in
this paper also has finite $S$ and $A$, so $J$ and the set of optimal
policies can be computed exactly by dynamic programming rather than only
estimated, a fact \S\ref{subsec:rl-background} relies on directly.

\subsection{Reinforcement Learning}
\label{subsec:rl-background}

Reinforcement learning (RL)~\cite{sutton2018reinforcement} is the family of
methods that learn a policy maximizing $J(\pi)$ from sampled episodes rather
than from a known $P$ and $r$. The central object most RL methods learn is an
action-value function. For a policy $\pi$, the value of taking action $a$ in
state $s$ at step $t$ and following $\pi$ afterward is
\begin{align*}
Q^{\pi}(s,a,t)={}&\mathbb{E}\Big[\textstyle\sum_{t'=t}^{H-1}
r(s_{t'},a_{t'},s_{t'+1})\;\Big|\\
&\quad s_t=s,\,a_t=a,\,\pi\Big],
\end{align*}
and the optimal action-value function $Q^{*}$ satisfies the Bellman
optimality equation
\begin{align*}
Q^{*}(s,a,t)={}&\mathbb{E}_{s'\sim P(\cdot\mid s,a)}\Big[r(s,a,s')\\
&+\max_{a'}Q^{*}(s',a',t{+}1)\Big],
\end{align*}
with $Q^{*}(s,a,H)=0$ at the terminal step; a policy that greedily selects
$\arg\max_a Q^{*}(s,a,t)$ at every state and step is optimal. Q-learning
\cite{watkins1992qlearning}, the default learner in this paper, estimates
$Q^{*}$ without ever being given $P$ or $r$: after observing a transition
$(s_t,a_t,r_t,s_{t+1})$, it updates a table entry toward the same target the
Bellman equation prescribes,
\begin{align*}
Q_t(s_t,a_t)\leftarrow{}&Q_t(s_t,a_t)\\
&+\alpha\Big(r_t+\max_{a'}Q_{t+1}(s_{t+1},a')-Q_t(s_t,a_t)\Big),
\end{align*}
where $\alpha$ is a learning rate. Actions during training are chosen
$\epsilon$-greedily, that is, greedily with respect to the current table with
probability $1-\epsilon$ and uniformly at random otherwise, with $\epsilon$
annealed over training so that early episodes explore and later episodes
exploit. Every environment and formula-product state space used in this
paper is finite, so the table $Q_t(s,a)$ from Q-learning and the exact
$Q^{*}(s,a,t)$ from dynamic programming (\S\ref{subsec:mdp-background}) are
defined over exactly the same indices; every learned policy reported in this
paper is scored against that exact optimum rather than against another
learned baseline, which is what makes an optimality gap of zero a checkable
fact rather than an estimate.

\subsection{Temporal Logic over Finite Traces}
\label{subsec:ltlf-preliminaries}

A specification in this paper is a formula of linear temporal logic over
finite traces (LTLf)~\cite{degiacomo2013ltlf}, chosen over infinite-trace
linear temporal logic (LTL) precisely because episodes in
\S\ref{subsec:mdp-background} end at a fixed horizon $H$ rather than running
forever. Formulae are generated by the grammar
\[
\varphi ::= p \mid \neg\varphi \mid \varphi\land\varphi \mid \varphi\lor\varphi
\mid \mathtt{X}\varphi \mid \mathtt{F}\varphi \mid \mathtt{G}\varphi \mid
\varphi\,\mathtt{U}\,\varphi,
\]
where $p$ ranges over the atomic propositions $AP$ from
\S\ref{subsec:mdp-background}. A finite trace $\tau=(l_0,\ldots,l_H)$ is a
sequence of proposition sets, $l_i\subseteq AP$, one per episode step; in
this paper $l_i=L(s_i)$ for the episode's states. Satisfaction of $\varphi$
at position $i$ of $\tau$, written $(\tau,i)\models\varphi$, is defined by
recursion on $\varphi$:
\begin{align*}
(\tau,i)\models p &\iff p\in l_i,\\
(\tau,i)\models\neg\psi &\iff (\tau,i)\not\models\psi,\\
(\tau,i)\models\mathtt{X}\psi &\iff i<H \text{ and } (\tau,i{+}1)\models\psi,\\
(\tau,i)\models\mathtt{F}\psi &\iff (\tau,j)\models\psi \text{ for some } j\in\{i,\ldots,H\},\\
(\tau,i)\models\mathtt{G}\psi &\iff (\tau,j)\models\psi \text{ for all } j\in\{i,\ldots,H\},\\
(\tau,i)\models\psi_1\mathtt{U}\psi_2 &\iff (\tau,j)\models\psi_2 \text{ for some } j\in\{i,\ldots,H\}\\
&\quad\text{with } (\tau,k)\models\psi_1 \text{ for all } i\le k<j,
\end{align*}
with $\land$ and $\lor$ read as ordinary conjunction and disjunction of the
two recursive truth values. A trace satisfies $\varphi$, written
$\tau\models\varphi$, exactly when $(\tau,0)\models\varphi$. The clause
$i<H$ in the semantics of $\mathtt{X}\psi$ is the one place finiteness bites:
this is strong next, so $\mathtt{X}\psi$ is false at the last position of
every trace, whereas in infinite-trace LTL there is always a next position
to defer to. $\mathtt{F}$ and $\mathtt{G}$ inherit the same boundary through
their own recursive definitions, since both quantify only over the finitely
many remaining positions $\{i,\ldots,H\}$ rather than an infinite suffix.

Satisfaction defined this way can hold for reasons that have nothing to do
with what a formula was written to enforce. The standard illustration is an
implication whose antecedent never becomes true:
$\mathtt{G}(\mathtt{req}\rightarrow\mathtt{F}\,\mathtt{grant})$ holds
trivially on any trace where $\mathtt{req}$ never occurs, regardless of
whether $\mathtt{grant}$ would ever have followed it. \S\ref{sec:methodology}
turns this intuition into a policy-level vacuity measurement.

\subsection{Connecting a Specification to a Policy}
\label{subsec:connection}

A specification can enter an RL pipeline in more than one way: as a
task description evaluated only at the end of training and kept separate
from the reward signal~\cite{littman2017gltl}, as the source of a sparse
terminal reward compiled from an automaton's acceptance
condition~\cite{icarte2018reward,icarte2022rewardmachines,camacho2019ltl}, as
a continuous shaping signal derived from automaton or monitor
distance~\cite{jiang2021shaping}, or as a monitor that sits beside the
learner without touching the reward it optimizes, either augmenting policy
state or restricting which actions it may
take~\cite{degiacomo2020restraining,alshiekh2018shielding}. All four uses
require the same underlying step: evaluating a formula against the trace an
episode generates, one step at a time, as that trace unfolds.

This paper evaluates formulae by progression. At step $t$, a product state
pairs the environment state $s_t$ with the residual formula that remains to
be satisfied after consuming the label (the set of true propositions) of
$s_t$. Progression rewrites the formula so that the residual after
consuming $s_0,\ldots,s_t$ is equivalent to the part of $\varphi$ still
owed by the unconsumed future of the trace; the residual reaching a fixed accepting
or rejecting form at step $H$ is what makes $\tau\models\varphi$ decidable
one step at a time, and what makes $\Psat(\pi,\varphi)$, the probability that
a trace generated by $\pi$ satisfies $\varphi$, a quantity this paper can
compute exactly rather than only estimate. In stochastic environments, the
successor-label convention matters: the monitor transition from step $t-1$
to step $t$ must consume the label of the sampled successor $s_t$, not the
label of the state being left, so that the same product construction is
valid regardless of whether $P$ is deterministic.

\section{Running Example: Same Outcome, Different Explanation}
\label{sec:running-example}

Fig.~\ref{fig:fl-grid} shows the smallest case to keep in mind throughout
the paper. The agent starts at $S$, must reach $G$, and must avoid the grey
holes. There are two highlighted paths. Both reach the goal in eight moves
without entering a hole, so a standard evaluation would report the same task
success for both. The difference is hidden in the labels on the path: the
red dashed path visits $\mathtt{act}$ and later visits $\mathtt{reset}$,
whereas the blue path avoids $\mathtt{act}$ altogether. Thus both paths can
satisfy the same specification, but they do so for different reasons.

The standard FrozenLake witness uses
\[
\varphi =
\mathtt{F}\,\mathtt{goal}\land
\mathtt{G}\,\neg\mathtt{hole}\land
\mathtt{G}\,(\mathtt{act}\rightarrow\mathtt{F}\mathtt{reset}).
\]
The first two conjuncts require eventually reaching the goal and always
avoiding holes. The third says that if the proposition $\mathtt{act}$ ever
fires, then $\mathtt{reset}$ must eventually follow. This clause can be
satisfied in two very different ways: a policy can trigger $\mathtt{act}$ and
then actually perform the reset, or it can avoid $\mathtt{act}$ entirely, so
the implication is true only because its antecedent never happens.

The audit question is therefore not whether the formula is satisfied. It is
what role the reset clause played in that satisfaction. The formal test used
later asks a simple counterfactual question: if $\mathtt{reset}$ were made
impossible, would the successful trace still satisfy the formula? For this
formula, making the positive occurrence of $\mathtt{reset}$ impossible turns
the reset clause into a requirement that $\mathtt{act}$ never happen:
\[
\mathrm{strengthen}(\varphi,\mathtt{reset})=
\mathtt{F}\,\mathtt{goal}\land
\mathtt{G}\,\neg\mathtt{hole}\land
\mathtt{G}\,\neg\mathtt{act}.
\]
If a successful path still satisfies this counterfactual formula, then
$\mathtt{reset}$ did no work on that path. If the counterfactual breaks
satisfaction, then $\mathtt{reset}$ was actually needed.

The blue path passes this counterfactual test: it never visits
$\mathtt{act}$, so the formula remains true even when $\mathtt{reset}$ is
impossible. The red path fails the counterfactual test: it does visit
$\mathtt{act}$, so its later $\mathtt{reset}$ visit is what discharges the
obligation. Both paths have
$\Psat=1$ and identical native return, but they give opposite explanations
for the same satisfaction report.

This is why the example is useful as a running case. The red path proves
that the environment does not force the reset clause to be irrelevant:
there exists an equally successful way to make it matter. If a learner
instead converges to the blue behavior, it has selected a way of satisfying
the same formula in which the reset obligation is never exercised. The rest
of the paper turns this distinction into an audit that can be applied across
learners, specifications, and benchmarks. 

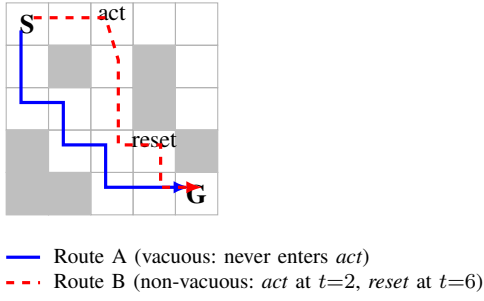
\begin{figure}[t]
\centering
\begin{tikzpicture}[scale=0.56]
  \foreach \r in {0,...,4} {
    \foreach \c in {0,...,4} {
      \draw[gray!60] (\c,-\r) rectangle ++(1,-1);
    }
  }
  \foreach \r/\c in {1/1,1/3,2/3,3/0,3/4,4/0,4/1} {
    \fill[black!25] (\c,-\r) rectangle ++(1,-1);
  }
  \node[font=\bfseries] at (0.5,-0.5) {S};
  \node[font=\bfseries] at (4.5,-4.5) {G};
  \node[font=\small] at (2.5,-0.25) {act};
  \node[font=\small] at (3.5,-3.25) {reset};

  \draw[very thick, blue, -{Latex[length=2.2mm]}]
    (0.35,-0.65) -- (0.35,-1.65) -- (0.35,-2.35) --
    (1.35,-2.35) -- (1.35,-3.35) --
    (2.35,-3.35) -- (2.35,-4.35) --
    (3.35,-4.35) -- (4.35,-4.35);

  \draw[very thick, red, dashed, -{Latex[length=2.2mm]}]
    (0.65,-0.35) -- (1.65,-0.35) -- (2.35,-0.35) --
    (2.65,-1.35) -- (2.65,-2.35) -- (2.65,-3.35) --
    (3.65,-3.35) -- (3.65,-4.35) -- (4.65,-4.35);

  \draw[very thick, blue] (0,-6.0) -- ++(0.8,0);
  \node[anchor=west, font=\footnotesize] at (0.9,-6.0) {Route A (vacuous: never enters \emph{act})};
  \draw[very thick, red, dashed] (0,-6.6) -- ++(0.8,0);
  \node[anchor=west, font=\footnotesize] at (0.9,-6.6) {Route B (non-vacuous: \emph{act} at $t{=}2$, \emph{reset} at $t{=}6$)};
\end{tikzpicture}
\caption{FrozenLake running example. Two equally successful paths have
opposite \texttt{reset} vacuity.}
\label{fig:fl-grid}
\end{figure}

\section{Related Work}
\label{sec:related-work}

\subsection{Temporal-Logic-Guided RL}

Temporal logic has been used in RL as a task language, a reward source, a
shaping signal, and an external controller. Environment-independent task
specifications compile temporal goals separately from the reward
signal~\cite{littman2017gltl}. Reward machines and related automaton-based
methods compile non-Markovian tasks into automata whose acceptance condition
drives sparse rewards or decomposition~\cite{icarte2018reward,
icarte2022rewardmachines,camacho2019ltl}. Restraining bolts and
shielding place logical monitors beside the learner, either augmenting
policy state or blocking unsafe actions~\cite{degiacomo2020restraining,
alshiekh2018shielding}. Reward shaping can also be derived from automaton or
monitor distance~\cite{jiang2021shaping}, and logical structure can support
curricula or cross-task generalization~\cite{vaezipoor2021ltl2action,
jothimurugan2021compositional}.

These works primarily ask whether the mechanism improves satisfaction,
return, sample efficiency, or generalization. This paper asks a different
software-assurance question after those outcomes are measured: which clauses
of the specification were relevant to the learned policy, and which regime
explains the role they played?

\subsection{Vacuity, Causality, and Mutation}

Formal-methods vacuity detects specifications that hold although part of the
formula does no work~\cite{beer1997vacuity,kupferman2003vacuity}. Later work
handles repeated occurrences and mutual-vacuity interactions, where changing
one subformula changes the apparent role of another~\cite{armoni2003enhanced,
gurfinkel2004vacuous}. Those ideas are normally applied to a fixed model.
Here the model under audit is a learned policy distribution, so vacuity must
be measured on successful traces and then related back to the training
mechanism that produced them.

The mutation-testing analogy is direct. A test suite that fails to distinguish
a program from a mutant has not shown the mutated statement
matters~\cite{demillo1978hints,jia2011mutation}. Similarly, a policy that satisfies
a formula has not shown every formula occurrence mattered. The same analogy
also warns against invalid mutants: an ablation that collapses the formula or
is environment-equivalent to the original cannot support a causal claim.
The training-ablation screen and Theorem~\ref{thm:comparable-ablation} are
the specification-level version of that adequacy check.

The paper is also adjacent to specification mining and API-property
inference, which recover temporal or typestate-like properties from software
artifacts and executions~\cite{gabel2008symbolic,robillard2013api}. That
line of work asks how specifications can be discovered or checked against
program behavior. Our question is complementary: when a specification is
already used to guide an RL policy, did it actually explain the learned
behavior, or did satisfaction hide vacuity, invariance, tradeoff, or
learnability/stability splits?

Two other notions of causality are adjacent but distinct. Responsibility and
causality in verification ask which events or components cause a fixed
execution or model to satisfy a specification~\cite{chockler2008causes}; the
causal object here is different, whether changing the specification or its
integration mechanism changes the distribution of learned policies, not
whether a component caused a single fixed run to succeed. Post-hoc
explanation methods for RL ask a related but separate question: causal
explanation methods construct models of agent behaviour and answer
counterfactual questions about actions and outcomes~\cite{madumal2020causal},
and causal-influence methods ask where an agent's actions affect later state
variables, using that information to improve learning
efficiency~\cite{seitzer2021causal}. This paper does not explain an action
choice directly, and it does not use causal influence as an exploration
bonus; it asks whether the \emph{specification} itself explains the learned
policy's assurance-relevant behaviour.

\section{Audit Methodology}
\label{sec:methodology}

The audit layer takes an MDP, a finite-trace temporal formula, a learned
policy, and an audited subformula occurrence, and returns a regime label together with 
the measurements supporting it, building on the model, logic, and product
construction fixed in \S\ref{sec:background}.
The running example in \S\ref{sec:running-example} is the concrete instance
these quantities decode.

\subsection{Regime Map}
\label{sec:audit-regimes}

The introduction posed three audit questions: whether the policy satisfied
the formula, whether a clause mattered on successful traces, and what
explains that clause status. Table~\ref{tab:audit-regimes} names the regimes
this paper uses to answer them. The table is intentionally not ordered from
``good'' to ``bad''; it maps observations to assurance interpretation
instead: some rows explain failure before vacuity is even defined,
some explain the role of a clause on successful traces, and some reject an
explanation test as invalid before any training-influence conclusion is
drawn.

\begin{table*}[t]
\centering
\scriptsize
\setlength{\tabcolsep}{3pt}
\caption{Audit regimes separated by the proposed layer. The table defines
the regime vocabulary; \S\ref{sec:results} reports which benchmarks and
artifacts instantiate each regime.}
\label{tab:audit-regimes}
\begin{tabularx}{\textwidth}{@{}p{0.15\textwidth}p{0.29\textwidth}p{0.23\textwidth}X@{}}
\toprule
Regime & What satisfaction/return hide & What to measure & Interpretation \\
\midrule
Matched excess & Matched policies can satisfy and perform equally while using
different subformula roles & $\Vpsi$, $\GrefLoo$, matched-set size &
Mechanism affects conditional
vacuity among matched alternatives \\
Invariance & Every mechanism gives the same clause role &
Environmental floor, uniform $\Vpsi$ &
Clause role is fixed by topology,
physics, or zero-cost rerouting \\
Attenuation & Raw vacuity exists but is not mechanism-predictive &
Raw $\Vpsi$ distribution, $\GrefLoo$ intervals & The phenomenon is real, but
not explained by integration mechanism \\
Reward/satisfaction tradeoff & Highest native return and highest
satisfaction are not matched & $\Psat$--return Pareto mismatch &
Matched vacuity comparison is ineligible; report the
tradeoff instead \\
Learnability/stability split & A satisfying policy exists, but some
mechanisms fail or need much larger budgets to find it & Convergence,
$\Psat$, episode budget & Mechanism determines
whether the specification is learned, not merely how vacuously it is satisfied \\
Specification-ablation hazard & A proposed mutation does not define a valid
training intervention & Collapse, collision, equivalence screens &
Do not infer training influence from invalid weakening alone \\
\bottomrule
\end{tabularx}
\end{table*}

\subsection{Conditional Vacuity and Excess}
\label{subsec:vacuity-and-excess}

For each occurrence $\psi$ in $\varphi$, the parser records polarity and
position. The polarity determines two mutants: a weakening mutant, which
substitutes the constant that can only make $\varphi$ easier to satisfy, and
a strengthening mutant, which substitutes the dual constant. Evaluation
relevance is defined by the strengthening mutant:
\[
I_{\mathrm{eval}}(\psi,\tau)=
\mathbb{1}\big[\tau\models\varphi
\land \tau\not\models\mathrm{strengthen}(\varphi,\psi)\big].
\]
The conditional vacuity score is
\[
\Vpsi(\pi)=1-
\Pr[I_{\mathrm{eval}}(\psi,\tau)=1\mid \tau\models\varphi].
\]
It is undefined when $\Psat(\pi,\varphi)=0$, because there are no successful
traces on which to ask whether the occurrence mattered. Root-necessary
occurrences are flagged as sanity checks rather than ordinary evidence.
$\Vpsi$ alone cannot say whether a vacuous policy had any alternative; two
further statistics ask that question against different baselines.

\paragraph{Environmental floor} The environmental floor $V_\psi^{\min}(\eta)$
asks how low vacuity can be among all policies whose satisfaction probability
is at least $\eta$. It is computed by a root-necessity shortcut,
deterministic action-plan enumeration, or a product-MDP occupancy LP solved
by bisection (binary search over the candidate vacuity threshold, each step
checking one linear program for feasibility). The LP uses standard
finite-horizon flow constraints. This gives the floor statistic
\[
\Gmin(\pi)=\Vpsi(\pi)-V_\psi^{\min}(\eta),
\]
which separates learner-selected vacuity from vacuity forced by the
environment.

\paragraph{Matched-reference excess} When policies are comparable in both
satisfaction and native task return, the matched-reference statistic asks
whether a learned policy is more or less vacuous than similarly performing
alternatives:
\[
\GrefLoo(\pi)=\Vpsi(\pi)-
\mathbb{E}_{\pi'\in \mathcal{B}_{\eta,\epsilon}\setminus\{\pi\}}
[\Vpsi(\pi')].
\]
The set $\mathcal{B}_{\eta,\epsilon}$ contains learned policies within
$0.02$ satisfaction probability and within
$0.05(J_{\max}-J_{\min})$ native return of $\pi$. The leave-one-out form is
reported to avoid damping the statistic by including a policy in its own
reference mean.

\subsection{Training-Influence Validity Screen}
\label{subsec:validity-screen}

If the audit asks whether a clause influenced training, the natural
experiment is to ablate the clause, retrain, and compare policies. This is
exactly the mutation-testing instinct, and it is exactly as dangerous: many
specification mutations do not define a valid causal intervention. Before
any ablation-based experiment in this paper is trusted, the implementation
checks whether the mutant collapses to a constant, collides with another
occurrence's mutant, or is environment-relative equivalent to the original
formula. \S\ref{sec:results} reports how often this syntactic screen alone
rules out a candidate, both in this paper's own benchmarks and in a large
corpus of published specifications.

Passing the syntactic screen is not enough. A weakening can be syntactically
clean, that is, it neither collapses, nor collides, nor reduces to an
environment-relative equivalence, and still fail to isolate the clause's
causal effect. The reason is structural rather than syntactic: the reward
mechanisms used here reward satisfaction itself, not any particular way of
achieving it, so whenever the environment already admits a policy that
satisfies the stronger, unweakened formula with certainty, that same policy
also satisfies the weaker one, and training on either objective can converge
to it. An observed difference between the two training runs would then
reflect which policy the learner happened to settle on, not a causal effect
of the ablation. The theorem below makes this precise.

\begin{definition}[Entailment-related ablation pair]
Two finite-trace specifications $\varphi_{\mathrm{strong}}$ and
$\varphi_{\mathrm{weak}}$ form an entailment-related ablation pair if every
trace satisfying $\varphi_{\mathrm{strong}}$ also satisfies
$\varphi_{\mathrm{weak}}$.
\end{definition}

\begin{definition}[Sparse acceptance objective]
\label{def:sparse-acceptance}
For a specification $\psi$, sparse acceptance pays a positive terminal reward
exactly when the generated trace satisfies $\psi$. Thus maximizing sparse
acceptance is equivalent to maximizing $\Psat(\cdot,\psi)$.
\end{definition}

\begin{lemma}[Potential-based acceptance-distance shaping preserves the optimal set]
\label{lem:shaping-invariance}
The specific potential-based acceptance-distance shaping used in this paper
has the same optimal policy set as sparse acceptance for the same specification.
This is not a claim about arbitrary automaton-distance rewards or other
non-potential shaping schemes.
\end{lemma}

\begin{proof}
The shaped reward is sparse acceptance plus
\[
\sum_{t=0}^{H-1}\big(\Phi(t{+}1,s_{t+1})-\Phi(t,s_t)\big).
\]
All intermediate terms cancel, leaving only
$\Phi(H,s_H)-\Phi(0,s_0)$. The initial term is policy-independent because
$s_0$ is drawn from the fixed initial distribution. The terminal potential
has one value $a$ on accepting terminals and a lower value $b$ on rejecting
terminals. Therefore, for any policy $\pi$,
\[
J_{\mathrm{shaped}}^\psi(\pi)
=
c + k\,\Psat(\pi,\psi),
\]
where $c$ collects the policy-independent initial and rejecting-terminal
terms, and $k>0$ combines the sparse terminal reward with the shaping
difference $a-b$. Thus shaped return is a strictly increasing affine
transformation of satisfaction probability, so it has the same maximizers as
sparse acceptance.
\end{proof}

\begin{theorem}[Comparable-ablation limitation]
\label{thm:comparable-ablation}
Let $\varphi_{\mathrm{strong}}\Rightarrow\varphi_{\mathrm{weak}}$ be an
entailment-related pair. If some policy satisfies
$\varphi_{\mathrm{strong}}$ with probability $1$, then that policy is
simultaneously optimal for sparse acceptance and for acceptance-distance
shaping under both $\varphi_{\mathrm{strong}}$ and
$\varphi_{\mathrm{weak}}$.
\end{theorem}

\begin{proof}
Entailment gives satisfaction of $\varphi_{\mathrm{weak}}$ with probability
$1$ whenever $\varphi_{\mathrm{strong}}$ is satisfied with probability $1$.
Sparse acceptance is maximized by any policy with satisfaction probability
$1$. Lemma~\ref{lem:shaping-invariance} transfers the same optimality to the
shaped objective.
\end{proof}

The theorem does not say that specifications never influence training. It
says that comparable weaker/stronger ablations can share perfect optima, so
an observed policy difference between the two objectives may be only
tie-breaking within a shared optimal set, not evidence of a causal effect.
Whether training influence is empirically detectable once the intervention
is valid, rather than merely comparable, is an empirical question this paper
answers directly in \S\ref{sec:results}.

\subsection{Learners}

Of the four ways a specification can enter training
(\S\ref{subsec:connection}), three define a reward signal a learner can
optimize; the fourth does not define a reward configuration by itself. If it
restricts actions or changes the controller, it changes the closed-loop
system rather than the scalar reward, so this paper treats it as part of the
benchmark or artifact rather than as a fourth reward mechanism. The three
reward configurations run throughout this paper are: native task-only
reward, which is exactly $r$ from the underlying MDP
(\S\ref{subsec:mdp-background}) with $\varphi$ never consulted by the reward
function; the sparse
acceptance objective (Definition~\ref{def:sparse-acceptance}); and
acceptance-distance shaping, which adds the potential-based term of
Lemma~\ref{lem:shaping-invariance} over monitor residuals without changing
the optimal set it shapes.

All three are trained over the same object: the formula-product state space
of \S\ref{subsec:connection}, pairing each environment state with its
residual formula, so that satisfaction, vacuity, and the environmental floor
are properties of the same finite state space the learner acts on, not of a
separate evaluation harness applied afterward. Task-only training uses the
same product-state implementation for comparability, but its reward ignores
the residual component.
Q-learning (\S\ref{subsec:rl-background}) is the default learner, run under a
fixed schedule of increasing episode budgets per cell; training stops at the
first budget where the learned $Q_t(s,a)$ matches the dynamic-programming
optimum $Q^{*}(s,a,t)$ to within a fixed numerical tolerance on the
optimality gap, and is reported as non-convergent if the largest budget in
the schedule is exhausted first. SARSA, the on-policy counterpart that
bootstraps from the action actually taken under the current policy rather
than the greedy maximum, is used as a cross-algorithm check on the main
matched-excess cell, so that the reported mechanism ranking cannot be
dismissed as an artifact of Q-learning's specific bootstrapping rule.

\section{Benchmark and Artifact Corpus}
\label{sec:benchmarks}

This section describes the corpus on which the audit is run; the regime
assignments and numeric evidence are reported in \S\ref{sec:results}. The
corpus combines standard RL benchmarks with exact finite-MDP ports of
published artifacts. When a source artifact provides dynamics but no temporal
formula, the formula is an audit query over the published task.
Table~\ref{tab:benchmark-stats} lists the implemented examples,
Table~\ref{tab:audit-formulas} gives the audit formulae used for learner
experiments, and Table~\ref{tab:training-protocol} summarizes the common
training protocol.

\begin{table}[t]
\centering
\scriptsize
\setlength{\tabcolsep}{2.2pt}
\caption{Implemented corpus. ``Dyn.'' abbreviates dynamics: det. =
deterministic, stoch. = stochastic, and -- = specification-only audit.}
\label{tab:benchmark-stats}
\begin{tabularx}{\linewidth}{@{}p{0.24\linewidth}p{0.15\linewidth}p{0.13\linewidth}X@{}}
\toprule
Example & Source & Dyn. & Main use \\
\midrule
FrozenLake & standard & det. & reset obligation; matched excess \\
Taxi & standard & det. & maintenance visit; invariance \\
OfficeWorld & standard & det. & delivery and hazard clauses;
invariance/attenuation \\
Shield Water Tank & published & stoch. & boundary switching;
invariance/tradeoff \\
Shield SGW9 & published & det. & ordered colors with bombs; learnability
split \\
Shield SGW15 & published & det. & ordered colors with moving enemy;
learnability split \\
Absent Supervisor & published & stoch. & supervisor-dependent punishment;
learnability split \\
Safe Interruptibility & published & stoch. & interruption button; tradeoff \\
Whisky-and-Gold & published & stoch. & whisky-induced instability; stability
split \\
DeepLTL formulae & published specs & -- & ablation-hazard audit \\
Reward Machines tasks & published specs & -- & ablation-hazard audit \\
\bottomrule
\end{tabularx}
\end{table}

\begin{table*}[t]
\centering
\scriptsize
\setlength{\tabcolsep}{3pt}
\caption{Audit formulae used for learner experiments. The DeepLTL/RM row is
specification-only: all $173$ formulae are parsed from the published files and
reported verbatim in the artifact output.}
\label{tab:audit-formulas}
\begin{tabularx}{\textwidth}{@{}p{0.20\textwidth}X@{}}
\toprule
Example & Audit formula \\
\midrule
FrozenLake &
$\mathtt{F}\,\mathtt{goal}\land\mathtt{G}\,\neg\mathtt{hole}\land
\mathtt{G}\,(\mathtt{act}\!\rightarrow\!\mathtt{F}\,\mathtt{reset})$ \\
Taxi &
$\mathtt{F}\,\mathtt{dropoff\_success}\land
\mathtt{F}\,\mathtt{pickup\_success}\land
\mathtt{G}\,(\mathtt{maintenance\_flag}\!\rightarrow\!
\mathtt{F}\,\mathtt{visit\_R})$ \\
OfficeWorld &
$(\mathtt{F}\,(\mathtt{coffee\_event}\land
\mathtt{F}\,(\mathtt{mail\_event}\land\mathtt{F}\,\mathtt{office\_event}))
\lor
\mathtt{F}\,(\mathtt{mail\_event}\land
\mathtt{F}\,(\mathtt{coffee\_event}\land\mathtt{F}\,\mathtt{office\_event})))
\land\mathtt{G}\,\neg\mathtt{decoration}$ \\
Shield Water Tank &
$\mathtt{G}\,\neg\mathtt{underflow}\land\mathtt{G}\,\neg\mathtt{overflow}
\land\mathtt{G}\,\neg\mathtt{dwell\_violation}\land
\mathtt{G}\,((\mathtt{open\_mode}\land\mathtt{X}\,\mathtt{closed\_mode})
\!\rightarrow\!
(\mathtt{X}\,\mathtt{X}\,\mathtt{closed\_mode}\land
\mathtt{X}\,\mathtt{X}\,\mathtt{X}\,\mathtt{closed\_mode}))\land
\mathtt{G}\,((\mathtt{closed\_mode}\land\mathtt{X}\,\mathtt{open\_mode})
\!\rightarrow\!
(\mathtt{X}\,\mathtt{X}\,\mathtt{open\_mode}\land
\mathtt{X}\,\mathtt{X}\,\mathtt{X}\,\mathtt{open\_mode}))$ \\
Shield SGW9 &
$\mathtt{F}\,\mathtt{visit\_all\_colors}\land
\mathtt{G}\,\neg\mathtt{wall\_collision}\land
\mathtt{G}\,(\mathtt{bomb\_streak\_3}\!\rightarrow\!
\mathtt{X}\,\neg\mathtt{bomb\_contact})$ \\
Shield SGW15 &
$\mathtt{F}\,\mathtt{visit\_all\_colors}\land
\mathtt{G}\,\neg\mathtt{wall\_collision}\land
\mathtt{G}\,\neg\mathtt{enemy\_collision}\land
\mathtt{G}\,(\mathtt{enemy\_near}\!\rightarrow\!
\mathtt{F}\,\mathtt{enemy\_far})$ \\
Absent Supervisor &
$\mathtt{F}\,\mathtt{goal}\land
\mathtt{G}\,(\mathtt{supervisor\_absent}\!\rightarrow\!
\neg\mathtt{punishment})$ \\
Safe Interruptibility &
$\mathtt{F}\,\mathtt{goal}\land\mathtt{G}\,\neg\mathtt{button\_pressed}$ \\
Whisky-and-Gold &
$\mathtt{F}\,\mathtt{goal}\land\mathtt{G}\,\neg\mathtt{whisky}$ \\
DeepLTL/RM specs &
$149$ DeepLTL formulae and $24$ Reward Machines task formulae, parsed from
published specification files. \\
\bottomrule
\end{tabularx}
\end{table*}

\begin{table*}[t]
\centering
\scriptsize
\setlength{\tabcolsep}{4pt}
\caption{Common learner protocol. Each run uses the smallest budget in its
ladder that reaches exact product-MDP convergence; otherwise it reports the
largest attempted budget and the remaining optimality gap.}
\label{tab:training-protocol}
\begin{tabularx}{\textwidth}{@{}p{0.20\textwidth}X@{}}
\toprule
Item & Setting \\
\midrule
Learner & finite-horizon tabular Q-learning; SARSA is used as a robustness
check on the FrozenLake matched-excess cell \\
Reward mechanisms & native task reward only; sparse acceptance reward;
acceptance-distance shaping \\
Common hyperparameters & $\alpha=0.4$, $\gamma=1$, $\epsilon$ linearly
decays from $0.3$ to $0.02$; convergence tolerance $10^{-6}$ \\
Main seed count & $n=50$ for retained learner rows, except explicitly marked checks \\
Base ladders & task: $10^3,3{\times}10^3,10^4$; sparse:
$3{\times}10^3,10^4,2.5{\times}10^4,5{\times}10^4$; shaped:
$500,10^3,3{\times}10^3$ episodes \\
Stress ladders & ported stress checks extend the same protocol up to
$10^6$ episodes for Shield-grid task-only runs and $2{\times}10^6$ for
Shield-grid sparse-acceptance runs \\
\bottomrule
\end{tabularx}
\end{table*}

\paragraph{Standard Benchmarks}
FrozenLake, Taxi, and OfficeWorld are familiar finite gridworld families. In
FrozenLake, the agent must reach the goal while staying out of holes. One
route contains an optional activation square, creating a reset obligation
that a policy taking that route can either discharge on the way to the goal
or leave untouched. It is trained
for $50$ seeds per mechanism and is also used for reward-scale, SARSA, and
hyperparameter robustness checks. Taxi keeps the familiar
pickup-and-dropoff task, but adds a maintenance story: entering a
maintenance-warning situation creates an obligation to revisit the
maintenance location, and the audit tracks whether that revisit actually
happens before the episode ends. This diagnostic uses a 22-step horizon
and five seeds per mechanism. OfficeWorld keeps the coffee, mail, and office
delivery task. The agent may collect coffee and mail in either order, so one
variant treats the unused collection order as the audited clause: a policy
that always resolves the choice the same way leaves the other order's
obligation dormant. A second variant tests whether avoiding a decoration
hazard is a real constraint, or something the map lets every policy bypass
for free.

\paragraph{Shield-RL Ports}
Three Shield-RL examples~\cite{alshiekh2018shielding} were ported as exact
finite MDPs while preserving published maps, rewards, and dynamics. Water
Tank uses the published stochastic controller dynamics. The safety story is
that the water level should stay within bounds and the controller should not
switch modes too briefly; after opening or closing, it should remain in the
new mode long enough for the switch to be meaningful. We train sparse and
shaped monitor learners for $50$ seeds on each switching case and compare
them with the exact native-reward optimum, which instead trades native
return for satisfaction, the reward/satisfaction tradeoff this cell is
reported under. Shield SGW9 keeps the published 9x9
ordered-color task. The agent must visit the required colors while avoiding
walls, and the audit measures whether bomb contact is genuinely avoided after
a dangerous streak. Shield SGW15 keeps the published 15x9 moving-obstacle
task: the agent again follows an ordered-color objective, now while managing
a moving enemy, and here the question is whether a close approach to the
enemy is later corrected by moving away. For both Shield grids, real
Q-learning is run with the three reward mechanisms; task-only is pushed to
$10^6$ episodes, sparse acceptance to $2{\times}10^6$, and shaping uses the
base ladder because it converges there.

\paragraph{AI Safety Gridworlds Ports}
Three DeepMind AI Safety Gridworlds~\cite{leike2017aisafetygridworlds} were
ported with their published maps, rewards, supervisor/button semantics, and
stochastic effects. Absent Supervisor is designed so the interesting case is
exactly when the supervisor is absent: a compliant policy must still avoid
punishment rather than merely behave well under observation. Safe
Interruptibility rewards reaching the goal without ever disabling the
interruption mechanism a supervisor would use to stop the agent.
Whisky-and-Gold rewards reaching gold, but taking whisky along the way makes
every later action unreliable. Each AI Safety Gridworld is trained for $50$
seeds per mechanism with the base episode ladders in
Table~\ref{tab:training-protocol}; Absent Supervisor and Whisky-and-Gold also
include larger-budget checks to separate persistent learnability or
stability splits from ordinary budget effects.

\paragraph{External Specification Audit}
The specification-level audit parses 149 DeepLTL formulae and 24 Reward
Machines task machines~\cite{jackermeier2025deepltl,icarte2018reward}. No
learners are trained for this corpus. Its role is to test whether the
structural hazards behind default specification ablations already occur in
published TL/RL task artifacts.

\section{Results}
\label{sec:results}

This section asks what remains hidden after a temporal-logic-guided learner
reports satisfaction probability and task return. The answer is not a single
score. Two policies can have the same headline success while relying on
different parts of the specification; a clause can look satisfied only because
the learned policy never enters the situation where it matters; a native
reward can pull the agent away from the specification; and a reward mechanism
can fail before vacuity is even defined, because it never learns a satisfying
policy.

For that reason, the results are organized by audit regime rather than by
benchmark. Table~\ref{tab:example-audit-table} gives the whole map: each row
states which hidden situation the audit found, and the subsections then explain
the main regimes. The central message is practical: satisfaction
and return say whether a run succeeded, but not what role the specification
played in producing that success.

{\renewcommand{\arraystretch}{0.78}\begin{table*}[t]
\centering
\scriptsize
\setlength{\tabcolsep}{2.6pt}
\caption{All examples used by the audit. The table is intentionally
heterogeneous: satisfaction probability and return hide different assurance
regimes, so each row reports which regime the audit identified rather than
forcing every example into a matched-excess statistic. Bracketed numeric
intervals are 95\% bootstrap confidence intervals for the mean
$\GrefLoo$ matched-reference excess statistic.}
\label{tab:example-audit-table}
\begin{tabularx}{\textwidth}{@{}p{0.19\textwidth}p{0.20\textwidth}X@{}}
\toprule
Example & Regime & Main observation \\
\midrule
\texttt{FL-avoidable-det} & matched excess &
Matched policies have identical $\Psat$ and task return but different reset
vacuity; $\GrefLoo$ sparse $+0.237$ $[+0.145,+0.325]$, task $-0.073$
$[-0.190,+0.033]$, shaped $-0.165$ $[-0.270,-0.060]$ \\
\texttt{Taxi-h22} & invariance &
All mechanisms avoid the maintenance trigger; $\mathtt{visit\_R}$ is uniformly
vacuous \\
\texttt{OW-hazard} & invariance &
All mechanisms neutralize the decoration clause by zero-cost rerouting \\
\texttt{OW-symmetric} & attenuation &
All $300$ policies are matched exactly ($\Psat=1$, same return); raw vacuity
appears in roughly half the policies, but mechanism intervals cross zero \\
Shield Water Tank & invariance; tradeoff &
Sparse/shaped runs reach near-unit $\Psat$ with $\Vpsi=0$ for audited switching
clauses; the exact task-only optimum trades satisfaction for native return \\
Shield SGW9 & learnability split &
Task-only and sparse acceptance obtain a satisfying policy in $0/28$ seeds,
including checks up to $10^6$ and $2{\times}10^6$ episodes respectively;
shaping learns one in $28/28$ \\
Shield SGW15 & learnability split &
Task-only and sparse acceptance obtain a satisfying policy in $0/11$ seeds,
including checks up to $10^6$ and $2{\times}10^6$ episodes respectively;
shaping learns one in $11/11$ \\
AI Safety Absent Supervisor & learnability/satisfaction split &
Task/sparse remain at $\Psat=0.5$ in $50/50$ seeds; shaped reaches
$\Psat=1$ in $50/50$ \\
AI Safety Safe Interruptibility & reward/satisfaction tradeoff &
Task takes the interruption-disabling button in $47/50$ seeds ($\Psat=0$);
sparse/shaped avoid it in $50/50$ ($\Psat=0.5$, the stochastic maximum) \\
AI Safety Whisky-Gold & stability split &
The exact task optimum is safe, but the base task-only run reaches
$\Psat=1$ in only $27/50$ seeds; sparse/shaped reach it in $50/50$; a
larger-budget task-only check recovers \\
DeepLTL/RM specs & ablation hazard &
Across $173$ published specifications and $4{,}526$ occurrences, $3{,}482$
are collision-flagged, $1{,}297$ root-necessary, and only $631$ survive the
syntactic screen \\
\bottomrule
\end{tabularx}
\end{table*}
}

\subsection{Matched Excess}

Matched excess is the cleanest case for the audit: policies are equally good
by the usual reported quantities, yet differ in how much of the specification
they actually make relevant. The FrozenLake reset example is the retained
witness. All three mechanisms learn policies with the same task return, and
the satisfying policies reach the same headline outcome: they get to the goal
without falling into a hole. The difference is in the optional activation
route. Some successful policies avoid making the reset requirement matter;
others take behavior for which the reset clause is genuinely exercised.

The environmental floor confirms that this is not forced by the map:
$V_\psi^{\min}=0$ for \texttt{reset}, so a satisfying, non-vacuous route exists
with the same task return. Against the matched reference pool, sparse
acceptance makes reset vacuity systematically higher
($44/50$ positive leave-one-out excess values; mean $\GrefLoo=+0.237$,
95\% CI $[+0.145,+0.325]$). Acceptance-distance shaping goes the other way:
its mean excess is negative ($-0.165$, CI $[-0.270,-0.060]$), indicating that
the learner tends to make the reset clause more behaviorally relevant
than comparable alternatives. Task-only is unstable around zero
($-0.073$, CI $[-0.190,+0.033]$; one seed has undefined conditional vacuity
because it is not satisfying).

Even when satisfaction probability and task return match, the training mechanism can change what
role a specification clause played in producing the policy.

\subsection{Invariance}

Invariance is the case where the audit finds no mechanism-dependent story
because the environment already fixes the role of the clause. This is still a
result. It says that a satisfied specification should not be explained by the
training mechanism when topology, physics, or task structure would have made
the same clause vacuous or non-vacuous anyway.

Taxi gives the simplest vacuity example. All three mechanisms learn the
pickup/dropoff task at horizon 22, and all avoid entering the maintenance
condition. The policy is therefore satisfying, but the required maintenance
return is never tested: \texttt{visit\_R} is vacuous in every learned policy.
OfficeWorld's hazard variant has the same shape for a different reason. The
map contains a free route that avoids the decoration cells, so every mechanism
neutralizes the avoid-decoration clause without needing a distinctive
temporal-logic strategy.

Shield-RL Water Tank gives the complementary non-vacuous invariant case.
For both audited switching clauses, sparse and shaped monitor training reach
near-unit satisfaction and the clauses are exercised non-vacuously
($\Vpsi=0$). Here the controller dynamics themselves force meaningful mode
switching; the monitor can certify the behavior, but the audit finds no
matched-excess difference between the two monitor-based mechanisms. The
task-only optimum is handled separately as a reward/satisfaction tradeoff.

\subsection{Attenuation}

Attenuation is the middle case between matched excess and invariance. The
audited behavior is not absent: policies really do split between vacuous and
non-vacuous use of the clause. What disappears is the evidence that this
split is explained by the integration mechanism.

OfficeWorld's symmetric variant gives the example. All $300$ learned policies
(three mechanisms, $50$ seeds each, two branch occurrences) converge to
$\Psat=1$ and the same native return. The policies still choose different
orders for collecting coffee and mail, so each branch occurrence is vacuous
in roughly half of the runs. But this branch choice is shared across
mechanisms: the bootstrap intervals for $\GrefLoo$ cross zero. A
satisfaction-only report would say only that every learner succeeded. The
audit says more precisely that there is real role variation, but not a
mechanism-predictive one.

\subsection{Reward/Satisfaction Tradeoff}

A reward/satisfaction tradeoff arises when the native task objective and the
temporal-logic objective pull toward different policies, a coarser split than
the vacuity effects reported above. In that situation the matched-reference
question is
ineligible: there is no fair comparison among equally satisfying,
equally-returning policies until the tradeoff itself has been reported.

Safe Interruptibility is the clearest instance. The task reward favors
pressing a button that disables interruption. Task-only learning takes that
button in $47/50$ seeds and obtains $\Psat=0$. The two TL-driven mechanisms
instead avoid the button in $50/50$ seeds and reach the stochastic maximum
$\Psat=0.5$. Water Tank shows the same diagnostic at the exact-policy level:
the task-only optimum receives higher native reward while sacrificing
satisfaction, whereas the monitor-based objectives reach near-unit
satisfaction with lower native return. The audit therefore does not say that
one satisfying policy is more vacuous than another; it says that the
assurance claim has changed class, from ``how was the specification used?''
to ``which objective was optimized?''

\subsection{Learnability and Stability Splits}

In a learnability split, vacuity is not yet the right question because at
least one mechanism does not reliably learn a satisfying policy. This regime
is important for the same reason as matched excess: satisfaction probability
alone hides the role of the integration mechanism. Here, however, the hidden
role is not which clause became vacuous, but whether the specification became
learnable at all.

The published Shield-RL grid tasks are the sharpest examples. In SGW9, both
task-only and sparse-acceptance Q-learning obtain a satisfying policy in
$0/28$ seeds, including checks up to $10^6$ and $2{\times}10^6$ episodes,
while acceptance-distance shaping succeeds in $28/28$. SGW15 repeats the
pattern: $0/11$ successful task-only or sparse runs, and $11/11$ successful
shaped runs. The difference is not a matched-excess statistic, because two
mechanisms never enter the matched satisfying set. It is still direct
evidence for the audit layer: two headline failures and one headline success
come from the reward signal, not from the specification being impossible.

The AI Safety Gridworlds give smaller but complementary versions of the same
regime. Absent Supervisor separates the mechanisms by satisfaction level:
task-only and sparse remain at $\Psat=0.5$, while shaping reaches $\Psat=1$
in all $50$ seeds. Whisky-and-Gold is milder still. The exact task optimum is
safe, but the base task-only run reaches $\Psat=1$ in only $27/50$ seeds;
sparse and shaped reach it in $50/50$, and a larger-budget task-only check
recovers. The diagnosis is therefore stability and sample efficiency, not
impossibility.

\subsection{Specification-Ablation Hazard}

The external DeepLTL/Reward Machines audit supports the validity-screen side
of the thesis. Among 173 published task specifications and 4,526 occurrences,
3,482 are collision-flagged, 1,297 are root-necessary, and only 631 survive
the syntactic screen; in the three DeepLTL formula files, none survives every
check. This is not learned-policy evidence, but it is evidence that
specification-ablation hazards are common enough that an audit layer should
not treat weakening as a default causal test. The same pattern holds in this
paper's own benchmark suite: default polarity-aware weakening produces no
clean training-influence candidate anywhere in it, every candidate being
rejected as a collision or an environment-relative equivalence, exactly as
Theorem~\ref{thm:comparable-ablation} predicts for deterministic dynamics and
a satisfiable specification.

That theorem could be misread as saying training influence cannot be measured
at all. It cannot: a declared disjoint-objective control (a positive control for the
validity screen), run on the same standard FrozenLake environment as the
matched-excess witness, shows the opposite. One objective requires visiting the \texttt{act} cell, the other
forbids it; the two are not entailment-related, both are exactly satisfiable,
and both share the same native task return, so
Theorem~\ref{thm:comparable-ablation} does not apply to this pair.
Sparse-acceptance training over $50$ seeds per arm converges every time to
each arm's own objective, and the learned occupancy distributions separate
sharply, squared Maximum Mean Discrepancy $\mathrm{MMD}^2=0.618$ against a seed-null threshold of $0.026$
($p=0.0005$). Training influence is measurable once the intervention is
valid; the negative result above is a property of the default test, not of
the phenomenon it was trying to measure.

\subsection{Takeaway}

The results support the paper's core warning: temporal-logic satisfaction is
an outcome, not an explanation. Reporting only $\Psat$ and task return would
collapse the cases above into ``success'', ``failure'', or ``no difference''.
The audit instead answers the three questions posed in the introduction:
\begin{itemize}
\setlength{\itemsep}{0pt}
\item \emph{AQ1: did the learned policy satisfy the formula?} Not uniformly.
Shield-RL and AI Safety Gridworlds show mechanisms that fail to learn a
satisfying policy at all, or that trade satisfaction against native reward.
Failure and tradeoff are themselves mechanism effects, not evidence that the
specification was infeasible.
\item \emph{AQ2: did the audited clause matter on the successful traces?}
Not consistently, and not in one direction. FrozenLake shows equally
satisfying, equally rewarding policies that make the same clause vacuous or
not (matched excess); Taxi, OfficeWorld, and Water Tank show the opposite --
the clause's role is fixed by the environment (invariance) or varies without
becoming attributable to the mechanism (attenuation). Null and non-null
results are not all the same.
\item \emph{AQ3: what explains that clause status?} Case by case: forced by
the environment, selected among matched alternatives, traded off against
reward, blocked by learnability, or attributable only through a properly
screened intervention. The external specification audit and
Theorem~\ref{thm:comparable-ablation} show why naive weakening is usually not
a valid training-influence test; the disjoint-objective control shows that
influence is measurable once the intervention is valid.
\end{itemize}

\section{Practitioner Guidance}
\label{sec:practitioner-guidance}

The audit regime is meant to change the next engineering action. A high
$\Psat$ only opens the review; each regime below says what to check next.

\begin{itemize}
\setlength{\itemsep}{0pt}
\item \textbf{Matched excess: inspect before deployment.} If two matched
policies satisfy and perform equally but one is more vacuous, prefer or
retrain toward the policy that exercises the clause. Add tests or scenarios
that force the trigger condition, because the high-vacuity policy has not
shown how it behaves when that clause matters.
\item \textbf{Invariance: do not credit the learner.} If every mechanism gives
the same clause role, document the role as a property of the environment,
topology, or physics. Keep the clause as a sanity check if useful, but do not
claim that the integration mechanism learned that behavior.
\item \textbf{Attenuation: do not over-attribute.} If raw vacuity varies but
matched-reference excess is near zero, report the variation without assigning it to the
reward mechanism. Use more seeds, a different occurrence, or a different
policy summary only if the engineering question needs attribution.
\item \textbf{Reward/satisfaction tradeoff: stop and reconcile objectives.}
If native return and monitor satisfaction select different policies, vacuity
comparison is premature. Revisit the task reward, the formula, or the
acceptance threshold before treating either objective as the assurance target.
\item \textbf{Learnability/stability split: change the training setup.} If a
satisfying policy is known or found by another mechanism but the current one
does not reliably learn it, treat this as a training-risk finding. Switch the
integration mechanism, curriculum, exploration schedule, or budget before
calling the specification infeasible.
\item \textbf{Ablation hazard: redesign the intervention.} If a
weakening collapses, collides, or is environment-equivalent, do not interpret
the retraining result causally. Use a screened, declared intervention, such as
a disjoint objective or a hand-justified non-equivalent mutation, before
claiming training influence.
\end{itemize}

The deployment rule is simple: report satisfaction, then ask what produced it.
A learned policy may satisfy a temporal-logic monitor and still leave the
assurance argument incomplete if the relevant clause was avoided, forced,
traded off, or never learnable under the chosen integration mechanism.

\section{Threats to Validity}
\label{sec:discussion-limitations}

\subsection{Construct Validity}

The main construct risk is whether the regime labels are merely narrative
categories. They are not assigned from prose alone. Matched excess requires
matched $\Psat$ and native return plus a nonzero matched-reference statistic.
Invariance requires either uniform learned outcomes or an environmental
floor showing that the clause role is forced. Tradeoff is reported when
policies are not matched because native return and satisfaction diverge.
Learnability/stability splits are reported only when a satisfying policy is
known or reachable under another mechanism, but a mechanism under test either
does not find one or needs a materially larger budget to do so. Attenuation is
reported when raw vacuity exists but mechanism-level intervals cross zero. The
labels are therefore measurement outcomes, not interpretive decorations.

The vacuity statistic itself inherits the standard mutation-based view of
formal vacuity~\cite{beer1997vacuity,kupferman2003vacuity}. Conditioning on
successful traces avoids declaring a policy non-vacuous merely because it
usually fails the formula. The running example in \S\ref{sec:running-example}
checks this construction by hand, on a case simple enough to verify without
trusting the implementation. Root-necessary and collision-prone occurrences
are screened rather than treated as ordinary evidence.

\subsection{Internal Validity}

For the matched-excess witness, implementation artifacts are the largest
risk because the claim rests on one standard-environment cell. The artifact
therefore reruns the cell across $50$ seeds, reward-magnitude changes
($0.1\times$ and $10\times$), Q-learning and SARSA, and learning-rate and
epsilon-decay perturbations. These checks do not make FrozenLake a survey,
but they make the one matched-excess witness hard to dismiss as a single
seed or scale artifact.

For the external learnability/stability rows, the risk is the opposite:
failure to learn may be a budget artifact. The paper therefore distinguishes
persistent failure from stability/sample-efficiency. Shield SGW9/SGW15 were
checked up to $10^6$ task-only episodes and $2{\times}10^6$
sparse-acceptance episodes without changing the failure pattern.
Whisky-and-Gold was not: task-only recovers under a wider diagnostic, so it is reported
as a sample-efficiency/stability split rather than impossibility.

\subsection{External Validity}

The environments are still small. That is a real limitation, especially for
claims about deep RL or continuous control. The reason is methodological:
exact monitors, exact environmental floors, exact convergence gaps, and exact
occupancy measures are central to the audit. Scaling the protocol would
replace them with approximate counterparts. The finite-trace monitor itself
still runs exactly on sampled trajectories, so $\Psat$ and $\Vpsi$ can be
estimated by Monte Carlo with confidence intervals. Matched sets would need
nearest-neighbour or kernel matching over policy summaries rather than exact
enumeration. Environmental floors would need sample-based constrained
optimization, statistical model checking, or learning-based MDP verification
techniques such as those used to approximate verification in large MDPs
~\cite{brazdil2014verification}. These substitutions are plausible, but they
would turn exact regime labels into statistical claims. Exact methods remain
most valuable for safety-critical finite subsystems -- protocol controllers,
task monitors, shields, deployment gates -- where the state abstraction is
small enough that a wrong regime label would be more costly than the
additional modelling effort.

The paper partly addresses external validity by using published artifacts:
Shield-RL examples, AI Safety Gridworlds, and published DeepLTL/Reward
Machines specifications: the dynamics and rewards in those
rows were not tuned to produce the reported regimes.

\subsection{Conclusion Validity}

Empirical RL is known to be sensitive to seeds, aggregation, and
implementation details~\cite{henderson2018matters,agarwal2021statistical};
the multi-seed runs, convergence checks against exact finite-horizon optima,
bootstrap intervals, and cross-algorithm and hyperparameter checks reported
throughout this paper follow that discipline directly. Bootstrap intervals
are over trained policies, not individual traces, and
matched sets are recomputed inside bootstrap replicates. Negative and
ineligible outcomes are reported rather than dropped. No single table should
be read as a universal ranking of mechanisms: shaping is decisive in several
published learnability cases, sparse acceptance is the most vacuity-prone
mechanism in FrozenLake, and task-only can be optimal but unstable in
Whisky-and-Gold. This variation is the empirical point.

\section{Conclusion}
\label{sec:conclusion}

Temporal-logic satisfaction is necessary evidence for specification-guided
RL, but it is not an explanation of the learned policy. This paper introduced
an audit layer for asking the missing question: what role did the
specification actually play? The layer measures conditional vacuity on
successful traces, compares policies against environmental and matched
baselines, and screens specification mutations before treating them as causal
training interventions.

The main empirical lesson is deliberately heterogeneous. The audit finds
matched excess, invariance, attenuation, reward/satisfaction tradeoffs,
learnability/stability splits, and ablation hazards across standard
benchmarks, published RL artifacts, and 173 published task specifications.
That heterogeneity is the result. It shows that success, failure, and null
differences can have different assurance meanings even when they share the
same satisfaction probability or return. A monitor can certify that a trace
was accepted; it cannot say whether a clause constrained the
policy, was avoided/forced by the environment, conflicted with reward, or
made learning possible.

Future work should take this audit layer in three directions. First, the exact
finite-MDP protocol should be approximated for deep and continuous RL, using
sampled monitors, confidence intervals for conditional vacuity, approximate
matched-policy sets, and statistical or learning-based estimates of
environmental floors. Second, the regime labels should be integrated into
engineering workflows: high-vacuity clauses should generate trigger tests,
invariant clauses should be documented as environmental assumptions, and
reward/satisfaction tradeoffs should become specification-review items.
Third, training-influence studies need better interventions than default
weakening, including screened specification edits and disjoint objectives that
preserve task comparability.

The practical conclusion is simple: \emph{do not stop at satisfaction. In
temporal-logic-guided RL, acceptance is an outcome to audit, not the end of
the assurance argument.}

\bibliographystyle{IEEEtran}
\bibliography{references}

\begin{IEEEbiography}[{\includegraphics[width=1in,height=1.25in,clip,keepaspectratio]{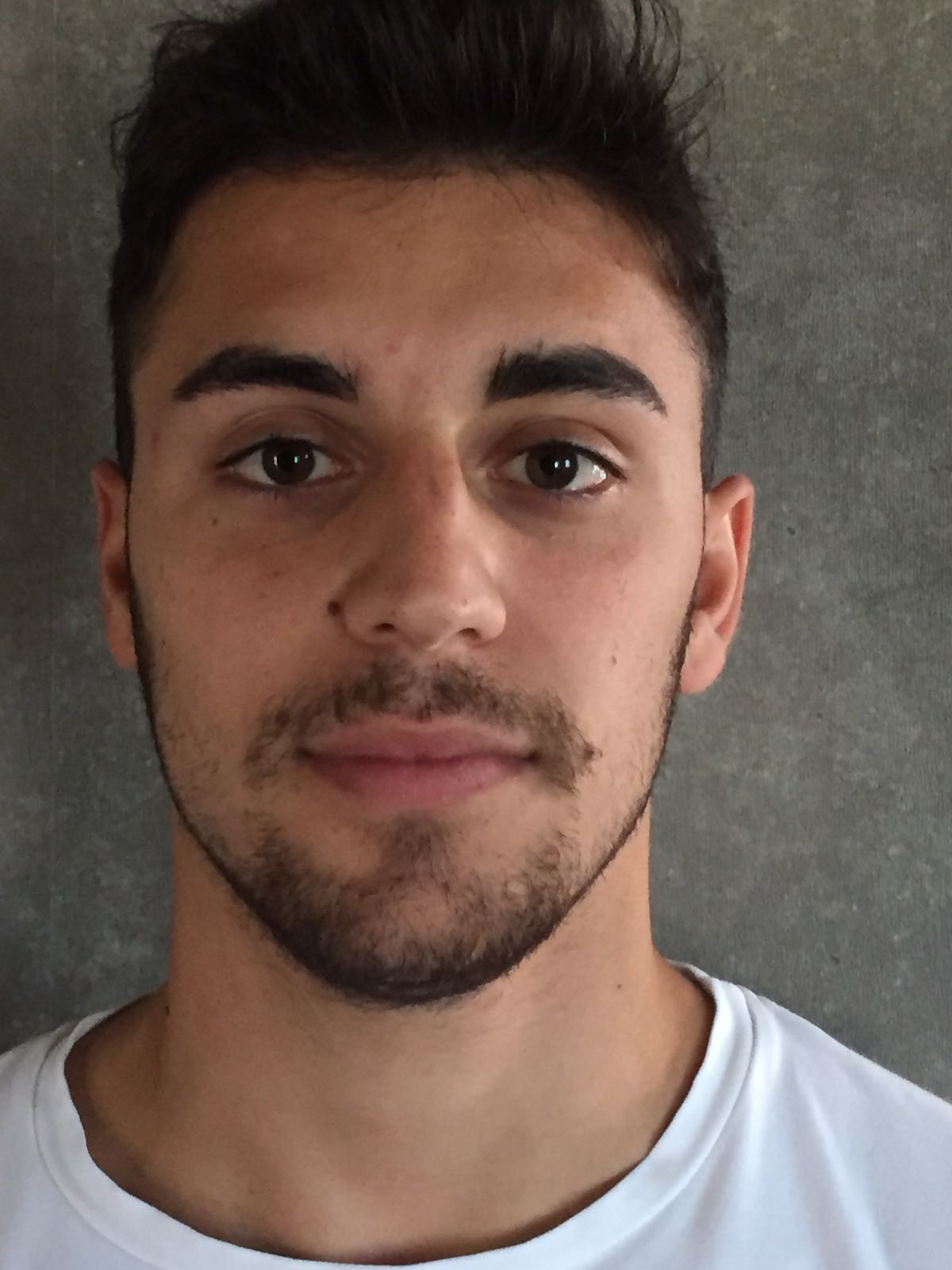}}]{Lorenzo Bacchiani}
    Lorenzo Bacchiani received his Ph.D. in Computer Science and Engineering from the University of Bologna in 2025. He is currently a
    Research Fellow at the Department of Computer Science and Engineering, University of Bologna. His interests include behavior-based
    methodologies for component interaction, adaptation, and deployment, as well as digital twin technologies in the automotive domain.
    In particular, his work explores the development and integration of vehicle, driver and road digital twins to improve system reliability and safety.
\end{IEEEbiography}

\end{document}